\documentclass[a4paper,12pt]{article}

\usepackage[english]{babel}
\usepackage{amsmath,amssymb}
\usepackage{times}
\usepackage{relsize}
\usepackage{graphicx}
\usepackage{url}
\usepackage[a4paper]{geometry}
\usepackage{booktabs}
\usepackage{multirow}
\usepackage{mparhack}
\usepackage{subfigure}
\usepackage{appendix}
\usepackage{authblk}
\usepackage{amsthm}

\usepackage{setspace}
\usepackage[noend]{algorithmic}

\usepackage{array}
\usepackage{cite}
\usepackage{eqparbox}
\usepackage{mdwmath}
\usepackage{stfloats}

\usepackage{epsfig}
\usepackage{xcolor}

\newtheorem{theorem}{Theorem}
\newtheorem{lemma}[theorem]{Lemma}

\newtheorem{corollary}[theorem]{Corollary}

\newenvironment{remark}[1][Remark]{\begin{trivlist}
\item[\hskip \labelsep {\bfseries #1}]}{\end{trivlist}}

\newenvironment{noindlist}
 {\begin{list}{\labelitemi}{\leftmargin=0em \itemindent=2.5em}}
 {\end{list}}

\def\CM{{\mathcal M}}
\def\CN{{\mathcal N}}

\title{On Tractability Aspects of Optimal Resource Allocation in OFDMA Systems}

\author[1]{Di Yuan}
\author[2]{Jingon Joung}
\author[2]{Chin Keong Ho}
\author[2]{Sumei Sun}
\affil[1]{{\small Department of Science and Technology, Link{\"o}ping University, Sweden}}
\affil[2]{{\small Institute for Infocomm Research (I$^2$R), A$^*$STAR, Singapore}}
\affil[ ]{\em{{\small Emails: diyua@itn.liu.se,  \{jgjoung; hock; sunsm\}@i2r.a-star.edu.sg}}}

\begin{document}

\date{}
\maketitle

\begin{abstract}
Joint channel and rate allocation with power minimization in
orthogonal frequency-division multiple access (OFDMA) has attracted
extensive attention. Most of the research has dealt with the
development of sub-optimal but low-complexity algorithms.
In this paper, the contributions comprise new
insights from revisiting tractability aspects of computing
optimum. Previous complexity analyses have been limited by assumptions
of fixed power on each subcarrier, or power-rate functions that locally
grow arbitrarily fast. The analysis under the former assumption does
not generalize to problem tractability with variable power, whereas
the latter assumption prohibits the result from being applicable to
well-behaved power-rate functions. As the first contribution, we
overcome the previous limitations by rigorously proving the problem's
NP-hardness for the representative logarithmic rate function. Next, we
extend the proof to reach a much stronger result, namely that the
problem remains NP-hard, even if the channels allocated to each user
is restricted to a consecutive block with given size. We also prove
that, under these restrictions, there is a special case with
polynomial-time tractability. Then, we treat the problem class where
the channels can be partitioned into an arbitrarily large but constant
number of groups, each having uniform gain for every individual
user. For this problem class, we present a polynomial-time algorithm
and prove optimality guarantee. In addition, we prove that the
recognition of this class is polynomial-time solvable.
% We finally discuss a network-flow-based algorithm for the general problem setting. Simulation results show that the proposed algorithm outperforms conventional algorithms of comparable or higher complexity.
\end{abstract}


%\newpage
\doublespacing

\section{Introduction}
\label{sec:introduction}

In orthogonal frequency division multiple access (OFDMA) systems,
resource allocation amounts to finding the optimal assignment of
subcarriers to users, and, for each user, the allocation of rate or
power over the assigned subset of sub-carriers. In this paper, we
focus on the problem of minimizing the total transmit power, subject to
delivering specified data rates, by power assignment and
subcarrier allocation. The popularity of OFDMA for wireless
communications has led to an intense research effort in this area. A
majority of the works has focused on heuristic and thus sub-optimal
solutions, see, for example,
\cite{WoChLeMu99,BoGrWoMe07,FeMaRe08,RhCi00,LiYaYa10,SaAnRa09,SEMoCi06},
and the references therein. For global optimality, a branch-and-bound
algorithm is developed in \cite{LiChSu09}.

We address a complementary but fundamental aspect of the resource
allocation problem: To what extent is it tractable? In contrast to the
significant amount of research on algorithms, works along the line of
tractability analysis are few \cite{GrBo06,HuGaKrSr10}.  The edge of
fundamental understanding of problem tractability is formed in respect
of the following limitations. In \cite{GrBo06}, the power on each
subcarrier is assumed to be given, i.e., it can be arbitrarily set to
any fixed value for the purpose of proving complexity. Thus the result
does not apply to the problem where powers are optimization
variables.  In
\cite{HuGaKrSr10}, the resource allocation problem is shown to be
NP-hard for one particular type of rate function $r_{mn}(P)\in
\mathcal{C}_{\text{inc}}$,
where $r_{mn}(P)$ is the rate as a function of power $P$ for the $m$th
user and $n$th subcarrier, and $\mathcal{C}_{\text{inc}}$ is the set
of all increasing functions such that $r_{mn}(0)=0$. By this result,
there exists some (but possibly ill-behaved\footnote{Indeed,
the proof in
\cite{HuGaKrSr10} relies on a power function (i.e., the inverse of the rate
function) growing arbitrarily fast for arbitrarily small rate
increase, meaning that the function is not locally Lipschitz
continuous.}) function in $\mathcal{C}_{\text{inc}}$ for which the
problem is NP-hard. However, the result does not carry over to
well-behaved subclasses of functions in $\mathcal{C}_{\text{inc}}$. In
fact, for any $r_{mn}(P)\in\mathcal{C}_{\text{linear}}$, where
$\mathcal{C}_{\text{linear}}$ is the set of linear functions, the
problem is solvable in polynomial time (see Section
\ref{sec:tractability}). Thus, whereas dealing with $\mathcal{C}_{\text{inc}}$
is in general intractable for the problem in question, there exists some subclass
in $\mathcal{C}_{\text{inc}}$ that admits global optimality at low complexity.

For the (wide) class of functions in $\mathcal{C}_{\text{inc}}$ but
not in $\mathcal{C}_{\text{linear}}$, the tractability of the resource allocation problem remains unknown and thus calls for
investigation. The aspect is of most relevance for the class of
increasing and concave functions $\mathcal{C}_{\text{concave}}$, for
which $\mathcal{C}_{\text{linear}} \subset
\mathcal{C}_{\text{concave}}
\subset \mathcal{C}_{\text{inc}}$ holds.
Indeed, it is very commonly assumed that the rate function is given by
$r_{mn}(P)=\log(1+g_{mn} P)$, with $g_{mn}>0$.  We will thus
investigate the complexity for this representative case of
$\mathcal{C}_{\text{concave}}$, and, if the problem is NP-hard for the
function, examine to what extent restrictions on input structure
(e.g., assuming the number of subcarrier for each user is part of the
input) will admit better tractability.  The specific contributions are
as follows.

%\begin{eqnarray*}{}
%\underbrace{\mathcal{C}_{\text{inc}} } &\supseteq \underbrace{\mathcal{C}_{\text{concave}} } \supseteq &\underbrace{\mathcal{C}_{\text{linear}}}. \\
%\text{NP hard} & \text{?} &  \text{polynomial-time}
%\underbrace{\mathcal{C}_{\text{linear}}} & \subseteq &  \underbrace{\mathcal{C}_{\text{inc}} }.\\
%\text{polynomial-time} & &  \text{NP hard}
%\end{eqnarray*}
%\begin{equation}\nonumber
%\underset{\text{polynomial-time}}{\underbrace{\mathcal{C}_{\text{linear}}}}  \subseteq \underset{\text{NP hard}}{\underbrace{\mathcal{C}_{\text{inc}}}}.
%\end{equation}

%\subsection{Contributions}

%The contributions are two-fold, focusing on optimum and close-to-optimal solutions. Firstly, we investigate aspects related to tractability and computation of global optimum solutions.
% Secondly, we propose a class of heuristic schemes that can outperform known solutions in the literature.
%We assume in the OFDMA system that
%each subchannel consists of one subcarrier; the results extend
%immediately to the case where each subchannel consists of multiple
%subcarriers with the same signal-to-noise ratio (SNR).

\begin{itemize}
\item We rigorously prove that given the representative logarithmic rate
function $r_{mn}(P)=\log(1+g_{mn} P)$ with $g_{mn}> 0$, the resource
allocation problem is NP-hard.  Since this rate function is in
$\mathcal{C}_{\text{concave}} $, it follows that in the hierarchy
$\mathcal{C}_{\text{linear}} \subset \mathcal{C}_{\text{concave}}
\subset \mathcal{C}_{\text{inc}}$, the hardness result holds
except for the linear case.  The contribution leads a significant
refinement of the result of tractability.

%\begin{eqnarray*}{RCC}
%\underbrace{\mathcal{C}_{\text{linear}}} &\subseteq &\underbrace{ \mathcal{C}_{\text{concave}} \subseteq  \mathcal{C}_{\text{inc}}}. \\
%\text{polynomial-time}& &  \text{NP hard}
%\end{eqnarray*}
%\begin{equation}\nonumber
%\underset{\text{polynomial-time}}{\underbrace{\mathcal{C}_{\text{linear}}}} \subseteq \underset{\text{NP hard}}{\underbrace{ \mathcal{C}_{\text{concave}} \subseteq  \mathcal{C}_{\text{inc}}}}.
%\end{equation}
%This result refines, rather than contradicts (see also our discussions earlier),
%the analysis in \cite{HuGaKrSr10}.

%
%\footnote{\cite{HuGaKrSr10} proved that there exists a rate function in $\mathcal{C}_{\text{inc}}$ leading to an NP-hard problem. This does not imply that {\em any} rate function in $\mathcal{C}_{\text{inc}}$ leads to an NP-hard problem; e.g., increasing linear rate functions admit polynomial-time complexity.},

\item We extend the NP-hardness analysis to arrive at
a much stronger result. Namely, the problem remains NP-hard, even if
the following two restrictions are {\em jointly} imposed: 1) the
subcarriers allocated to each user form a consecutive block, and 2)
the block size is given as part of the problem's input.  We also prove
that, with these two restrictions, there is a special case admitting
polynomial-time tractability.

\item We identify a tractable problem subclass, with
the structure that the channels can be partitioned into an {\em
arbitrarily many but constant} number of groups, possibly with varying
group size, and the channels within each group have uniform gain for
each user (but may differ by user). The original problem is, in fact,
equivalent to having the number of groups equal to the number of
subcarriers. The tractable subclass goes beyond the logarithmic
rate function -- the result holds as long as the single-user resource
allocation is tractable. Moreover, we prove that recognizing this
problem class is tractable as well.

%\item We present a network-flow-based algorithm. In each iteration, the
%algorithm performs a generalized optimal matching.
%%The number of channels in the matching acts as an input-adaptive control parameter.
%%We prove that the algorithm is optimal for the tractable problem class with one channel group.
%%Simulations, using integer programming in benchmarking, are provided to show the efficiency of the algorithms in approaching high-quality solutions for the general case.
%Simulation results show that the proposed algorithm outperforms conventional algorithms of comparable or higher complexity.
\end{itemize}

The remainder of the paper is organized as follows. The system model
is given in Section~\ref{sec:system}. In
Section~\ref{sec:tractability}, we provide and prove the base result
of NP-hardness.  Section~\ref{sec:tractabilityrestriction} is devoted
to the problem's tractability with restrictions on channel allocation.
In Section~\ref{sec:tractableclass}, we consider the problem class
with structured channel gain and prove its tractability. In addition,
we prove that recognizing this problem class is computable in
polynomial time.
% The network-flow-based algorithm is detailed in Section \ref{sec:flow}.
%We report and analyze simulation results in Section \ref{sec:simulation}, and
Conclusions are given in Section~\ref{sec:conclusion}.

\section{System Model}
\label{sec:system}

Consider an OFDMA system with $M$ users and $N$ subcarriers.  In this
paper, the terms subcarrier and channel are used interchangeably.  For
convenience, we define sets $\CM = \{1, \dots, M\}$ and $\CN =
\{1, \dots, N\}$. Notation $r_{mn}(P_{mn})$ is reserved for the rate as a function of the transmission power $P_{mn}\geq 0$
for the $m$th user and $n$th subcarrier. The inverse function $r^{-1}_{mn}$,
returning the power for supporting a given rate, is denoted by
$f_{mn}$.  The required rate of user $m$ is denoted by $R_m$.

The forthcoming analysis focuses on the representative increasing, concave rate function $r_{mn}(P_{mn}) =
\log_2(1+g_{mn} P_{mn})$, where $g_{mn}>0$ represents the channel gain (normalized such that the noise variance is one).
Hence, the corresponding signal-to-noise ratio (SNR) is
$g_{mn}P_{mn}$. Some of the tractability results generalize to
$\mathcal{C}_{\text{inc}}$; we will explicitly mention the
generalization when it applies.

The optimization problem is to minimize the total power by joint
channel and rate allocation, such that the users' required rate
targets are met. The problem is formulated formally below.  Throughout
the rest of the paper, we refer to the problem as minimum-power
channel allocation (MPCA).

\begin{noindlist}
\item[{\bf Input}:] User set $\CM = \{1, \dots, M\}$ and
channel set $\CN = \{1, \dots, N\}$ with $M \leq N$, positive
channel gain $g_{mn}, m \in \CM, n \in \CN$, and positive rate targets
$R_m, m \in \CM$.

\item[{\bf Output}:] A channel partitioning represented by $(\CN_1, \dots, \CN_M)$,
where $\CN_m \subset \CN, m \in \CM$, $\CN_{m_1} \cap \CN_{m_2} = \emptyset$,
for all $m_1, m_2 \in \CM, m_1
\not=m_2$, and non-negative power $P_{mn}$, $m \in \CM, n \in \CN_m$,
such that $\sum_{n \in \CN_m} \log_2(1+P_{mn}g_{mn}) \geq R_m, m
\in \CM$ and the total power $\sum_{m \in \CM} \sum_{n \in \CN_m}
P_{mn}$ is minimized.
\end{noindlist}

By the rate function, there is a unique mapping between rate
allocation and power expenditure. In addition, at optimum, the rate may
be zero on some of the allocated channels of a user. Given the channel
allocation, power optimization is determined by solving the
single-user rate allocation problem by water-filling \cite{Bi90} in
linear time (e.g.,
\cite{FeMaRe08}). The combinatorial nature of the multi-user problem
stems from the fact that the users can not share a subcarrier, and the
core of problem-solving is channel partitioning.

\section{Tractability of MPCA: Base Results}
\label{sec:tractability}

To motivate the investigation of tractability in view of the current
literature, two remarks are noteworthy. First, an NP-hard
problem may become tractable by imposing restrictions to the structure
of its input parameters (e.g., the shape of the objective function).
%For example, for the well-known problem of
%scheduling in ad hoc networks, the complexity, assuming that
%the channel gain can behave arbitrarily, was proved in
%\cite{BjVaYu04}. For the same problem with geometric channel gain, however, the
%complexity remained open until the formal proof in
%\cite{GoOsWa07}. The structure and shape of the cost function also have strong
%impact on complexity.
An example is the traveling salesman problem (TSP) having a cost
function with the so called Klyaus-matrix structure (meaning that
inversed triangular distance inequality holds). In this case, TSP is
solvable in polynomial time \cite{BuDeVaVaWo98}.  For MPCA, it
is in fact polynomially solvable, if the power-rate function would be
linear (see the end of this section). As the second
remark, an NP-hard problem may become tractable by removing some of
its constraints (and hence enlarging the solution space). For example,
minimum spanning tree (MST) with constrained node degree is NP-hard,
but becomes easily-solved if this constraint is removed.

%Before discussing the complexity results, a few remarks on
%complexity are noteworthy. In combinatorial optimization, the number
%of possible combinations of discrete choices is typically exponential
%in the input size. Yet some of problems admit time-efficiently
%algorithms guaranteeing global optimum. Consider for example the
%traveling salesman problem (TSP) in a complete graph of $n$ nodes, and
%matching in a bipartite graph of $2n$ nodes. Both have $n!$ potential
%solutions. Whereas the former is NP-hard, the latter can be solved
%easily using, for example, the well-known Hungarian method.

It has been widely accepted that MPCA and other related OFDMA resource
allocation problem are difficult. However, to the best of our
knowledge, no formal analysis other than the results in
\cite{GrBo06,HuGaKrSr10} is available. The study in \cite{GrBo06}
formalizes the NP-hardness result with the assumption that the powers
on all channels are given.  This is equivalent to introducing
constraints fixing the power values.  By the second remark above, the
result does not answer the tractability if these constraints are
removed (that is, the original problem with variable power). Indeed,
for any linear rate function as well as the problem class in Section
\ref{sec:tractableclass}, the NP-hardness proof in \cite{GrBo06}
remains valid for fixed power, but these problem classes with variable
power are solvable in polynomial time. For the analysis in
\cite{HuGaKrSr10}, the proof requires unbounded power growth for
arbitrarily small rate increase. Specifically, for the power function
$f$, $f(n+\frac{1}{k})$ is $k$ times higher than $f(n)$, for arbitrary
positive integers $n$ and $k$.  This assumption of ill-behaved
power-rate function excludes not only $f(x)= 2^x-1$ (the inverse of
the logarithmic rate function), but also all locally Lipschitz
continuous functions.  Recall that a (not necessarily continuous)
function $f$ is locally Lipschitz continuous, if for any $x$, there
exists a small real number $\epsilon$ and an arbitrarily large but
constant real number $D$, such that $|f(x+\epsilon)-f(x)|
\leq D \epsilon$, that is, the growth of the function is bounded when
the change in the input diminishes.  Clearly, $f$ is not locally
Lipschitz continuous, if $f(n+\frac{1}{k})$ increases by factor $k$
over $f(n)$ for any $n$ and arbitrarily large $k$. Hence, by the
previous remark of the impact of cost function on tractability, the
tractability under more well-behaved functions calls for investigation.

We provide the tractability results that overcome the limitations of
the currently available analysis in two aspects. First, we present a
rigorous proof of the problem's NP-hardness with the representative
logarithmic rate function. Second, in the next section, we extend the
proof to reach a much stronger result, stating that the problem
remains NP-hard even with two heavy restrictions on channel
allocation.

\begin{theorem}
\label{theo:nphard}
MPCA, as defined in Section \ref{sec:system}, is NP-hard.
\end{theorem}
\begin{proof}
As the proof is rather technical, we outline the basic idea
and defer the details to Appendix A. The proof uses a reduction from
3-satisfiability (3-SAT). Two groups of users are defined. At optimum,
each user in the first group either uses one channel of superior channel gain,
or splits the rate on three inferior channels, but not
both. This corresponds to the true/false value assignment in
3-SAT. The optimal power for this group of users is a constant,
while the optimal power for the second user group gives the correct
answer to 3-SAT. 
\end{proof}

\begin{corollary}
\label{theo:nphardrate}
MPCA remains NP-hard, even if
the rate requirements of the users are uniform.
\end{corollary}
\begin{proof}
Follows immediately from the equal-rate values used in the proof of
Theorem \ref{theo:nphard}.
\end{proof}

Earlier in this section, it was claimed that MPCA with any linear rate
function is tractable. Even if this case is not much of practical
interest, it is instructive in showing the importance of
input assumption on problem tractability.

\begin{theorem}
\label{theo:linear}
MPCA with linear rate function $r_{mn}(P) = \ell_{mn}P$, where $\ell_{mn}
\geq 0$, $m \in \CM, n \in \CN$, is solvable in polynomial time.
\end{theorem}
\begin{proof}
Since the rate (and hence power) function is linear, it follows that
for any user $m$, it is optimal to allocate the entire rate $R_m$ to a
single channel.  Specifically, denoting by $\CN_m$ the channel set
allocated to $m$, the optimal selection is the channel giving $\min_{n
\in \CN_m} 1/\ell_{mn}$. Because of this structure at optimum,
MPCA reduces to pairing the $M$ users with $M$ out of the $N$
channels.  Hence the problem is equivalent to a minimum-weight
matching problem (also known as minimum-cost assignment
\cite{AhMaOr93}) in a bipartite graph with node sets $\CM \cup \{M+1,
\dots, N\}$ and $\CN$; the former represents the augmentation of $\CM$
by $M-N$ artificial users. For edge $(m,n)$, with $m \in \CM$ and $n
\in \CN$, the cost is $1/\ell_{mn}$. All edges adjacent to artificial
users have zero cost.  Because the assignment problem is
polynomial-time solvable, the theorem follows.
\end{proof}

\begin{remark}
In light of Theorem~\ref{theo:nphard},
Corollary~\ref{theo:nphardrate}, and Theorem~\ref{theo:linear}, we
remark on the significance of assumption of input on tractability, by
revisiting the result provided in \cite{GrBo06}.  In this reference,
the result is proven for fixed power, that is, the power of each
channel is part of problem input.  Under this condition,
\cite{GrBo06} provides an elegant hardness proof of a reduction from the number
partitioning problem, by setting specific power values on the
channels.  As long as power is fixed, a line-by-line copy of the proof
in \cite{GrBo06} remains valid even if the underlying MPCA rate function
is linear. A similar observation applies to the tractable
problem class that will be detailed in Section
\ref{sec:tractableclass}.
In conclusion, for the original MPCA problem having power allocation
as part of the output (for which the key assumption of \cite{GrBo06}
does not apply), our analysis provides new insights in the
tractability rather than contradicting the previous results. $\Box$
\end{remark}

\section{Tractability with Restrictions on Channel Allocation}
\label{sec:tractabilityrestriction}

Consider imposing jointly two restrictions to channel allocation.
First, the number of channels to be allocated to each user is
given. Tractability under this restriction is of significance to
two-phase OFDMA resource allocation (see
\cite{KiLiLi03}) that determines the number of channels per user in
phase one, followed by channel allocation in phase two.  The second
restriction is the use of consecutive channels, that is, the channels
of every user must be consecutive in the sequence $1,\dots,N$; this
channel-adjacency constraint has been considered in, for example,
\cite{WoOtMc09}. We prove that MPCA remains NP-hard even with these two seemingly
strong restrictions, although there is a special case admitting
polynomial-time tractability.

\begin{theorem}
\label{theo:nphardrestriction}
MPCA remains NP-hard, even if
the number of channels allocated to each user, i.e., the cardinality
of $\CN_m, m=1, \dots, M$, is given in the input, and $\CN_m, m=1,
\dots, M$ must contain consecutive elements in the channel sequence
$1,\dots,N$.
\end{theorem}
\begin{proof}
As we prove the result via an augmentation of reduction
proof of in Appendix \ref{sec:appa}, the details are deferred to Appendix
\ref{sec:appb}.
In brief, the proof is built upon introducing a set of additional
channels in the MPCA instance in Appendix \ref{sec:appa}, in a way
such that, at optimum, the channels allocated to each user are
consecutive and the corresponding cardinality is known.
\end{proof}

Consider further narrowing down the problem to the case where $\CN$
consists of $M$ uniform-sized subsets of consecutive channels, and $N$
is a multiple of $M$. For this case, the problem is tractable, as
proven below.

\begin{theorem}
\label{theo:special}
If $N$ is divisible by $M$, and it is restricted to allocate exactly
$\frac{N}{M}$ consecutive channels to each user, then MPCA is tractable
with time complexity $O(\max\{M^3, MN\})$.
\end{theorem}
\begin{proof}
We prove the result by a polynomial-time transformation to the
well-solved matching problem in a bipartite graph.  For the problem
setting in question, channel partition is unique - the $M$ subsets
created by the partition are $\{1, \dots, N/M\}$, $\{N/M+1,\dots,
2N/M\}$, $\dots$, $\{N-N/M+1, \dots, N\}$, each containing $N/M$
channels. The bipartite graph has $2M$ nodes representing the users
and the subsets of channels. For each pair of the two node groups, say
user $m$ and channel subset $S$, there is an edge of which the cost
equals the power of meeting the user rate using the channels in $S$.
This cost is computed in $O(\frac{N}{M})$ time (by single-user rate
allocation). The complexity of computing all the edge costs is hence
$O(MN)$. Maximum-weighted matching in a bipartite graph of $V$ nodes
and $E$ edges is solved in $O(V^2\log V + VE)$ time
\cite{AhMaOr93}. In our case, the second term is dominating, giving a
time complexity of $O(M^3)$, which completes the proof.
\end{proof}

Theorem \ref{theo:special} provides a generalization of the trivial
case of $M=N$. For $M=N$, solving MPCA amounts to finding an optimal
matching with complexity $O(M^3)$. Moreover, from the proof, it
follows that the analysis generalizes to any rate function in
$\mathcal{C}_{\text{inc}}$, except that the complexity of single-use
rate allocation has to be accounted for accordingly. The observation
yields the following corollary.

\begin{corollary}
\label{theo:specialinc}
The time required for computing the optimum to the MPCA problem class in Theorem \ref{theo:special} is 
of $O(\max\{M^3, M^2 T({\frac{N}{M}})\})$, where
$T({\frac{N}{M}})$ denotes the time complexity of optimal rate
allocation of a single user on $\frac{N}{M}$ channels.
\end{corollary}

\section{A Tractable Problem Class}
\label{sec:tractableclass}

Denote by $K$ a (possibly large) fixed positive integer independent of
$M$ or $N$.  Let $\mathcal{K}=\{1,2,\cdots,K\}.$ Consider MPCA in
which the channels can be partitioned into (at most) $K$ groups where
every user in the same group has the same channel gain, but the
channel gains differ by user. Thus, we can write $\CN =
\bigcup_{k\in \mathcal{K}}
\CN_k$, such that for any channel $n \in \CN_k$, the channel 
gain depends only on the user index $m$, i.e.,
$g_{mn}=g_m$. Equivalently, the rate functions belong to the subclass
that satisfy $r_{mn}=r_m, n \in \CN_k, k\in\mathcal{K}, m\in\CM.$ Note
that the channel groups may vary in size, and the channel gain still
differs by user within each group.  In the sequel, we refer to the
problem class as $K$-MPCA.  The problem class is justified by
scenarios with $K$ distinct bands and channel difference is
overwhelmingly contributed by the separation of the bands in the
spectrum, whereas the subcarriers with each band are considered
invariant for each user.

Consider $1$-MPCA. The  problem structure is significantly simpler
than the general case.  Namely, the optimization decision is
no longer which, but how many channels each user should use.
% It follows that, for any subset of users $\CM \subset \CM'$, the optimum allocation of $h$ channels (with $h \geq |\CM'|$) among the users in $\CM'$ is independent of channel allocation of the rest the users.
It follows that, for any subset of users $\CM' \subset \CM$, the
optimum allocation of $h$ channels (with $h \geq |\CM'|$) among the
users in $\CM'$ is independent of channel allocation of the rest the
users.  Thus $1$-MPCA exhibits an optimal substructure, i.e., a part
of the optimal solution is also optimal for that part of the
problem. The observation leads to a dynamic programming line of
argument for problem-solving. As proven below, the solution strategy
guarantees optimality in polynomial time.

\begin{theorem}
\label{theo:1mpca}
Global optimum of $1$-MPCA can be computed by dynamic programming in $O(MN^2)$ time.
\end{theorem}
\begin{proof}
Consider the partial problem of optimally allocating $h$ channels to
the users in $\{1,\dots,m\}$ with $h \geq m$, and $c_m(h)$ the corresponding optimum
power. Clearly, at the optimum of this subproblem, the number of
channels of user $m$ is an integer in the set $\{1, 2, \dots,
h-m+1\}$. (The upper bound $h-m+1$ corresponds to having $m-1$
channels left for the other $m-1$ users.) Allocating $k \in \{1, 2,
\dots, h-m+1\}$ channels to user $m$, the power equals $p_m^k +
c_{m-1}(h-k)$, where $p_m^k$ denotes the power of user $m$ with $k$
channels. This gives the following recursive formula for computing the
optimal number of channels for user $m$.
\begin{equation}
\begin{split}
\label{eq:dp}
c_m(h) = \displaystyle \min_{k=1,\dots,h-m+1} \{p_m^k + c_{m-1}(h-k)\}
\end{split}
\end{equation}
We arrange the values $c_m(h)$ for $m=1,\dots,M$ and $h=1, \dots, N$
in an $M \times N$ matrix.  Entries corresponding to infeasible
solutions are called invalid, and their values are denoted by
$\infty$. In the matrix, $c_m(h)=\infty$ for all entries where $h<m$,
or $h > N-M+m$. We compute the valid entries as follows.  For the
first row, computing the entries $c_1(1), \dots, c_1(N-M+1)$ in the
given order are straightforward, and each entry requires $O(1)$
computing time.  Next, entries $c_m(m)$, i.e., one channel per user
for the first $m$ users, are calculated in $O(M)$ time for
$m=1,\dots,M$.  The bulk of the computation calculates the remaining
entries row by row, starting from row two. For row $m$, the
computations follow the order $c_m(m+1),\dots,c_m(N-M+m)$. Each of
these entries is calculated using formula \eqref{eq:dp}.  For the
valid entries of a row, the total number of comparisons that they are
used for computing the next row is $1+\dots+N-M+1$.  Hence the
complexity for computing row $m, m=2,\dots,M$, is of $O(N^2)$, and the
overall time complexity is of $O(MN^2)$. The last entry computed,
$c_M(N)$, gives the optimal allocation of the $N$ channels to the $M$
users and hence solves $1$-MPCA. In parallel, the solution is stored
in a second matrix of same size. Solution recording clearly has
lower complexity than $O(MN^2)$, and the theorem follows.
\end{proof}

In the following theorem, we generalize the dynamic programming
concept to any positive integer $K$. The
generalized algorithms are able to solve $K$-MPCA to global optimality.

\begin{theorem}
\label{theo:kmpca}
Global optimum of $K$-MPCA can be computed by dynamic programming in $O(MN^{2K})$ time.
\end{theorem}
\begin{proof}
Let $N_j = |\CN_j|, j=1, \dots, K$. For the first $m$ users, denote
by $c_m(h_1, \dots, h_K)$ the optimum power of allocating $h_j$
channels of group, where $h_j \leq N_j$, $j=1,\dots, K$.  Denote by
$p_m^{(k_1,\dots,k_K)}$ the power for user $m$, if it is allocated
$k_j$ channels of channel group $j$, $j=1,\dots,K$. We introduce the
convention that $p_m^{(0,\dots,0)} = \infty$ for convenience.  By
enumerating user $m$'s allocation of channels of the $K$ groups, we
obtain the following recursion formula for $c_m(h_1, \dots, h_K)$.

%\vspace{-0mm}
%\begin{align}
\begin{equation}\label{eq:dpk}
\begin{split}
& c_m(h_1, \dots, h_K)=  \min_{k_j \in \{0,\dots,\min\{h_j, N-M+1\}\}, j=1,\dots,K} \\
& \{p_m^{(k_1,\dots,k_K)} + c_{m-1}(h_1-k_1, \dots,h_K-k_K)\}  \end{split}
\end{equation}
%\end{align}

Extending the algorithm in the proof of Theorem \ref{theo:1mpca}, the
corresponding matrix for $K$-MPCA has dimension $M \prod_{j=1,\dots,K}
N^j$, which does not exceed $O(MN^K)$. To compute an entry, there are
no more than $O(N^K)$ calculations (including addition and comparison)
using $\eqref{eq:dpk}$.  For each calculation, the time required to
compute $p_m^{(k_1,\dots,k_K)}$ is linear\footnote{This result follows by directly applying the single-user rate assignment
to $K$ channel groups, where the channels in each group have
uniform gain.} in $K$.  Since $K$ is a constant, this computation does
not add to the complexity. These observations lead to the overall
complexity of $O(MN^{2K})$. Finally, entry $c_M(N^1, \dots, N^K)$ is
clearly the optimum to $K$-MPCA.  The proof is complete by observing
that, similar to $1$-MPCA, recording the channel allocation solution
does not form the computational bottleneck.
\end{proof}

\begin{remark}
The polynomial-time tractability of $K$-MPCA holds only
if $K$ is not dependent on $M$ or $N$. In fact, the general setting of
MPCA is equivalent to $N$-MPCA, i.e., $N$ channel groups with single
channel each. The dynamic programming algorithm remains applicable for
$K=N$. From Theorem \ref{eq:dpk}, however, the algorithm corresponds
to enumerating the solution space, and the running time is
exponential. $\Box$
\end{remark}

Having concluded the tractability of $K$-MPCA, a natural question to
ask next is whether or not the identification of the problem class
is tractable as well.  Theorem~\ref{theo:recog} states that
this is indeed the case.

\begin{theorem}
\label{theo:recog}
Recognizing $K$-MPCA can be performed in $O(MN)$ time for $K=1$, and
in $O(MN^2)$ time for any $K \geq 2$.
\end{theorem}
\begin{proof}
For $K=1$, identifying the problem class simply amounts to verifying,
for each user, whether or not all the $N$ channels are of the same
gain; this immediately leads to the $O(MN)$ time complexity
result. For $K \geq 2$, we construct a graph ${\cal G}$. The graph
has $N$ nodes, each representing a channel in $\CN$.  Consider two arbitrary
channels $n_1, n_2 \in \CN, n_1 \not= n_2$.  If the two channels have
the same gain for each of the users, i.e., $g_{mn_1} =
g_{mn_2}$, $\forall m \in \CM$, we denote it by $n_1 \cong
n_2$. Checking whether or not this is the case runs obviously in
$O(M)$ time. If $n_1 \cong n_2$, we add edge $(n_1, n_2)$ to ${\cal
G}$. Doing so for all unordered channel pairs has time complexity
$O(MN^2)$. Next, note that the equivalence relation of channels is
transitive, i.e., if $n_1 \cong n_2$ and $n_2 \cong n_3$, then
$n_1 \cong n_3$. Hence channels that are equivalent for all users form
a clique in ${\cal G}$, whereas channels that differ in gain for at least one
user are not connected in ${\cal G}$. Consequently the number of
strongly connected components in ${\cal G}$ equals the number of
channel groups, each of which contains channels being equivalent for
any user. Identifying the number of strongly connected components
requires no more than $O(N^2)$ time for ${\cal G}$. Therefore the
bottleneck lies in the $O(MN^2)$ complexity of obtaining the graph,
and the theorem follows.
\end{proof}

\begin{remark}
The tractability results of this section are not restricted to the
specific rate/power function defined in the section of system
model. The problem class remains tractable (although the overall
complexity may grow) for any function in ${\mathcal C}_{\text{inc}}$,
as long as the function admits polynomial-time rate allocation of
single user. $\Box$
\end{remark}

\section{Conclusions}
\label{sec:conclusion}

We have considered the OFDMA resource allocation problem of minimizing
the total power of channel allocation, so as to satisfy some rate
constraints. Although it has been known that assuming the most general
(and ill-behaved) increasing rate functions leads to NP-hard problems, we
have shown that the same conclusion holds even if we restrict the
class to increasing and concave rate functions. Interestingly, the
problem admits a polynomial-time solution if the rate function is an
increasing linear function. Hence, progress in the fundamental
understanding on the tractability of the problem is made in the
following sense: we have sharpened the boundary of tractability to
between increasing concave and increasing linear rate functions.
Finally, we have also identified specific cases when the problem
remains NP-hard, or admits polynomial-time solutions, under various
restrictions.

\appendices
\def\theequation{A\arabic{equation}}
\setcounter{equation}{0}
\section{Proof of Theorem \ref{theo:nphard}}\label{sec:appa}

There is no doubt that MPCA is in NP. The NP-hardness proof uses a
polynomial-time reduction from the 3-satisfiability (3-SAT) problem
that is NP-complete\cite{GaJo79}.  A 3-SAT instance consists in a
number of boolean variables, and a set of clauses each consisting of a
disjunction of exactly three literals. A literal is either a variable
or its negation. The output is a yes/no answer to whether or not there
is an assignment of boolean values to the variables, such that all the
clauses become true.  Denote by $v$ and $w$ the numbers of variables
and clauses, respectively. For any binary variable $z$, its negation
is denoted by $\hat z$. For the proof, we consider 3-SAT where each
variable and its negation together appear at most 4 times in the
clauses. Note that 3-SAT remains NP-complete with this restriction
\cite{To84}. Without loss of generality, we assume that each variable
$z$ appears in at least one clause, and the same holds for its
negation ${\hat z}$, because otherwise the optimal value of the
variable becomes known, and the variable can be discarded. Hence the
total number of occurrences of each literal in the clauses is between
one and three.

We construct an MPCA instance with $M=2v+w$ and $N=7v+w$. We
categorize the users and channels into groups, and, for convenience,
name the groups based on their roles in the proof. The users consist
in $2v$ literal users and $w$ clause users. The channels are composed
by three groups: $v$ super-channels, $6v$ literal channels, and $w$
auxiliary channels. The rate target $R_m = 1.0, \forall m \in \CM$.

\begin{figure}[t!]
\centering
\subfigure[Users and channels for a variable.]{
\includegraphics[scale=0.7]{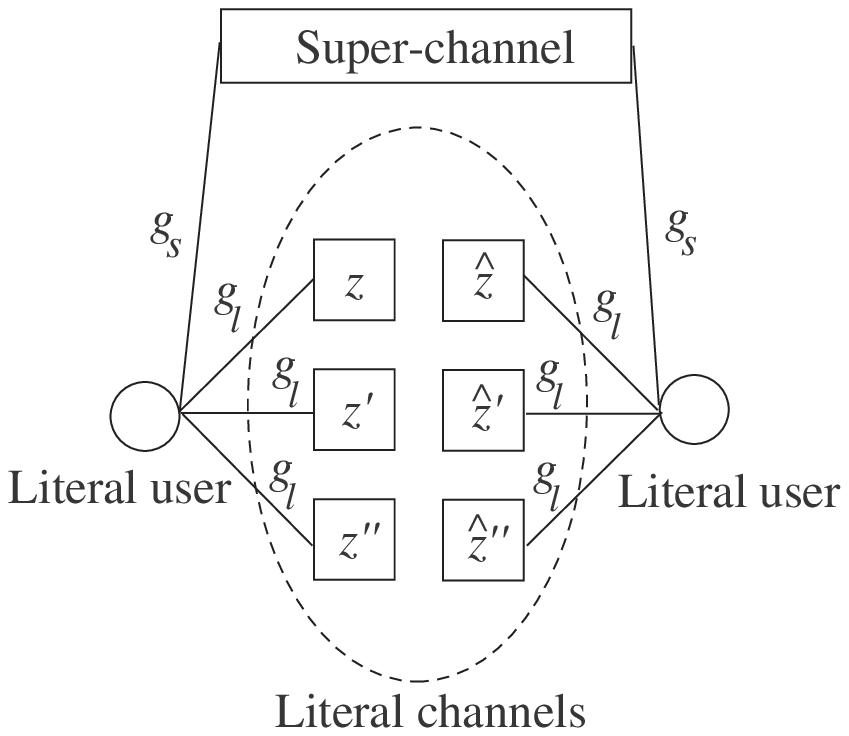}
\label{fig:literal}
}
\hspace{0.5em}
\subfigure[User and channels for a clause.]{
\includegraphics[scale=.7]{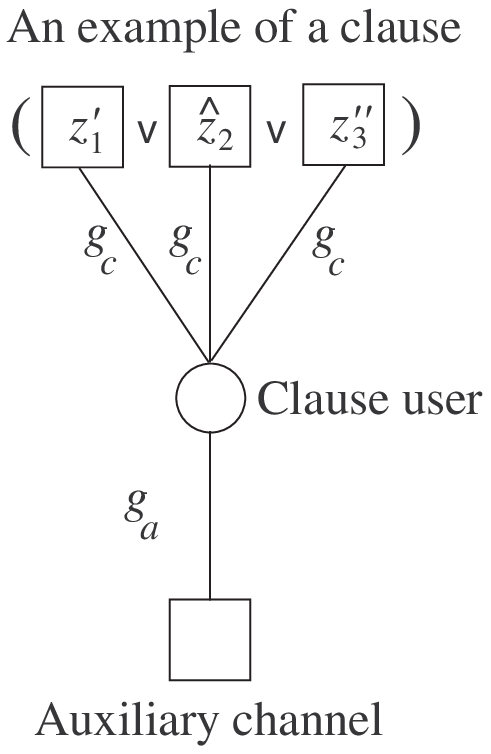}
\label{fig:clause}
}\vspace{-2mm}
\caption{An illustration of problem reduction.}
\label{fig:reduction}\vspace{-4mm}
\end{figure}

Problem reduction is illustrated in Fig. \ref{fig:reduction}.
For each binary variable $z$, three identical literal channels,
denoted by $z$, $z'$ and $z''$, are defined. A similar construction is
done for $\hat z$. For this group of six literal channels, one
super-channel is defined.  We introduce two literal users for the
seven channels.  One user has channel gain $g_l$ on the three literal
channels $z$, $z'$, and $z''$, and the other, complementary literal
user has channel gain $g_l$ on the remaining three literal
channels. Both users have channel gain $g_s$ on the super-channel. See
Fig. \ref{fig:literal}. Next, recall that each literal appears at most
three times in the clauses in the 3-SAT instance. In the proof, for
any binary variable $z$ appearing $t$ times in total in the clauses,
with $1 \leq t \leq 3$, the occurrences are represented by any $t$
elements in $\{z, z', z''\}$ in any order. A similar representation is
performed for the negation $\hat z$. For each clause, we introduce one
clause user with gain $g_c$ on the channels corresponding to the
original literals in the clause. In addition, one auxiliary channel is
defined per clause user with channel gain $g_a$.  See Fig. \ref{fig:clause}. We set $g_s = g_c = 1$, $g_a =
\frac{1}{0.9w + 0.1}$, $g_l = \frac{g_a} {26} = \frac{1}{26
\cdot (0.9w+0.1)}$. For the user-channel combinations other than those specified,
the channel gain is $g_{\epsilon} =
\frac{1}{53w}$. For each user, we refer to the four channels
with gain higher than $g_\epsilon$ as valid channels, and the other
$7v+w-4$ channels with gain $g_\epsilon$ as invalid channels. From
the construction, clearly the reduction is polynomial.

We provide several lemmas characterizing the optimum to the
MPCA instance. The first three lemmas
use the following optimality conditions of single-user rate allocation (e.g.,
\cite{FeMaRe08}).  First, for any user, the derivatives of the power
function, evaluated at the allocated rates, are equal on all channels
with positive rates. Second, for channels not used, the function
derivatives at zero rate are strictly higher than those of the used
channels. For $f(x) =
\frac{2^x-1}{g}$, where $x$ is the rate allocated and $g$ is
the channel gain, the derivative $f'(x) = \ln(2) \frac{2^x}
{g}$.

\setcounter{theorem}{0}
\begin{lemma}
\label{theo:epsilon}
There is an optimum allocation in which no user is allocated
any invalid channel.
\end{lemma}
\begin{proof}
Suppose that at optimum a clause user $m_1$ is allocated at least one
invalid channel. Assume $m_1$ is also allocated any valid channel,
then the invalid channels carry zero rate, because putting the entire
rate of 1.0 on the auxiliary channel, the function derivative is at
most $\ln(2)
\cdot \frac{2}{g_a} =2\ln(2)(0.9w+0.1)$, whereas the function derivative
for any invalid channel, at zero rate, is $\ln(2) \cdot 53w >
2\ln(2)(0.9w+0.1)$. Thus the invalid channels can be eliminated from
the allocation of $m_1$. Assume now the auxiliary channel is allocated
to another user, say $m_2$. By construction, the auxiliary channel of
$m_1$ is invalid for $m_2$. Consider re-allocating any invalid channel
of $m_1$ to $m_2$, and allocating the auxiliary channel to $m_1$.
Clearly, the total power will not increase. At this stage, the
remaining invalid channels allocated to $m_1$ carry zero flow.
Repeating the argument, we obtain an optimal allocation
in which no clause user is allocated any invalid channel.

Let $m_1$ be any literal user and suppose it is allocated one or more
invalid channels at optimum. If $m_1$ is allocated any of its four
valid channels, the function derivative at rate 1.0 is at most $\ln(2)
\cdot 52 (0.9w + 0.1) < \ln(2) \cdot 53w$, and therefore the invalid
channels carry zero rate and can be removed from $m_1$'s allocation.
Assume therefore all four valid channels are allocated to other users.
Consider any literal channel of $m_1$, and suppose it is allocated to
user $m_2$. For $m_2$, this literal channel is an invalid one.
Swapping the allocation of the literal channel and any invalid channel
currently allocated to $m_1$, the total power will not grow, and the
remaining invalid channels allocated to $m_1$ can be released.
The lemma follows from applying the procedure repeatedly.
\end{proof}

\begin{lemma}
\label{theo:super}
If a literal user is allocated its super-channel in the optimal
solution, then none of the three literal channels is allocated to the
same user.
\end{lemma}

\begin{proof}
Putting the entire rate of 1.0 on the
super-channel, the derivative value is $2\ln(2)$. For any literal
channel, the derivative at zero rate is $\frac{\ln(2)}{g_l}$. That
$g_l = \frac{1}{26 \cdot (0.9w+0.1)}$ and $w \geq 1$ lead to
$\frac{\ln(2)}{g_l} > 2\ln(2)$, and the result follows.
\end{proof}

\begin{lemma}
\label{theo:literal}
If a literal user is not allocated its super-channel in the optimal
solution, then the user is allocated all the three literal channels
\end{lemma}

\begin{proof}
By the assumption of the lemma and Lemma \ref{theo:epsilon}, the
literal user in question is allocated one, two, or all three of its
literal channels, with total power $f_1 = \frac{1}{g_l}$, $f_2 = 2
\cdot \frac{2^{1/2}-1}{g_l}$, and $f_3 = 3 \cdot
\frac{2^{1/3}-1}{g_l}$, respectively. Note that,
in the latter two cases, it is optimal to split the rate evenly
because of the identical gain values. Clearly, $f_3 < f_2 < f_1$.

We prove that cases one and two are not optimal. Suppose that, at optimum,
the literal user is allocated two of the literal channels. Then the
remaining literal channel is used to carry a strictly positive amount
of flow of a clause user. Consider modifying the solution by allocating
all the three literal channels to the literal user, and letting the
clause user use the auxiliary channel only.  The power saving for the
literal user is exactly $f_2-f_3 > \frac{0.04}{g_l} > \frac{1}{g_a} $,
whereas the power increase for the clause user is less than
$\frac{1}{g_a}$. This contradicts the optimality assumption.
Hence case two is not optimal. Since $f_1-f_2 > f_2-f_3$, a similar
argument applies to case one, and the result follows.
\end{proof}

By Lemmas \ref{theo:super}--\ref{theo:literal}, at optimum, a literal
user will use either the super-channel only, or all the three literal
channels. Hence, for any literal in the 3-SAT instance, either none or
all of the corresponding three literal channels become blocked for the
clause users. Consequently, there is a unique mapping between a
true/false variable assignment in the 3-SAT instance and the
availability of literal channels to the clause users in the MPCA
instance. The total power consumption of all the literal users equals
exactly $v + 78v
\cdot (2^{1/3}-1)(0.9w+0.1)$ at optimum.
In the remainder of the proof, we concentrate on the power consumption
of the clause users.

\begin{lemma}
\label{theo:max}
If every clause user is allocated at least one of the three literal
channels corresponding to the literals in the clause in the 3-SAT instance,
then the total power for all clause users is at most $w$.
\end{lemma}

\begin{proof}
Allocating the entire rate of 1.0 to one literal
channel gives a power consumption of $\frac{2^1-1}{g_c}
= 1$ for the clause user. As there are $w$ clause users, the lemma follows.
\end{proof}

\begin{lemma}
\label{theo:min}
If at least one clause user is not allocated any of its three literal channels,
the total power for all clause users is strictly higher than
$w$.
\end{lemma}

\begin{proof}
By the assumption, at least one clause user is allocated the
auxiliary channel only, with power $\frac{1}{g_a}$. Each of the other
$w-1$ clause users is allocated at most four channels.  Assuming the
availability of all four channels and setting $g_a = g_c$ leads to an
under-estimation of the power consumption. The under-estimation has
a total power of $4(w-1)\frac{2^{1/4}-1}{g_c} +
\frac{1}{g_a}$ $= 4 (w-1) (2^{1/4}-1) + (0.9w+0.1)$ $> 0.4w - 0.4
+ 0.9w + 0.1 = 1.3w -0.3 \geq w$.
\end{proof}

By Lemmas \ref{theo:max}-\ref{theo:min}, the optimum power for the
clause users is at most $w$ if and only if the answer is yes to the
3-SAT instance. Thus the recognition version of MPCA
is NP-complete, and its optimization version is NP-hard.

\def\theequation{A\arabic{equation}}
\setcounter{equation}{0}
\section{Proof of Theorem \ref{theo:nphardrestriction}}\label{sec:appb}

In the proof of Theorem \ref{theo:nphard}, a 3-SAT instance of $v$
variables and $w$ clauses is reduced to an MPCA instance with $2v+w$
users and $7v+w$ channels, such that no user will be allocated more
than three channels at optimum.  We make an augmentation
by adding $2v$ channels, which we refer to as dummy
channels. The channel gain of the dummy channels equals $g_\epsilon$
(defined in the proof of Theorem
\ref{theo:nphard}) for all users. After the augmentation, there is a total of $9v+w$ channels, organized
in three blocks.  The sequence of channels is as follows. The first
block has $3v$ channels, including the $v$ super-channels and the $2v$
dummy channels, in a sequence of $v$ chunks of $3$ channels each. Each
chunk is composed by one super-channel and two dummy channels.  The
next block has the $6v$ literal channels, with $v$ chunks having $6$
channels each. Every chunk corresponds to a binary variable $z$ in the
3-SAT instance, and the six literal channels appear in the order $z,
z', z'', {\hat z}, {\hat z}'$ and ${\hat z}''$. The third block
contains the $w$ auxiliary channels. Consider the resulting MPCA
instance with the restriction that, for the literal and clause users,
respectively, the numbers of channels allocated per user are
three and one.  In addition, channel allocation must be consecutive in
the given sequence.

To prove the hardness result, consider first a relaxation of the
problem, in which the two restrictions of channel allocation are
ignored for the literal users. For this relaxation, it is clear that
Lemmas \ref{theo:epsilon}-\ref{theo:literal} remain valid.  By Lemma
\ref{theo:epsilon} and the signal-channel restriction of the clause
users, each of these users will be allocated one of the four valid
channels. Obviously, the optimum is to allocate one literal channel,
or the auxiliary channel if all the three literal channels are allocated
to literal users. Thus, as long as at least one literal channel is
available to every clause user, the result of Lemma \ref{theo:max}
holds. In addition, the validity of Lemma \ref{theo:min} obviously
remains. Therefore the optimum to the problem relaxation provides the
correct answer to the 3-SAT instance.

By Lemma \ref{theo:literal}, in the optimum of the relaxed problem,
each literal user is either allocated its three consecutive literal
channels, or the super-channel.  In the latter case, we modify the
solution by allocating the two dummy channels accompanying the
super-channel, without changing the total
power. After the modification, the allocation satisfies the
cardinality requirement and restriction of using consecutive channels
for all literal users, and the theorem follows.

\onehalfspacing
\bibliographystyle{plain}

\end{document}